\documentclass[letterpaper, 10 pt, conference]{ieeeconf}
\IEEEoverridecommandlockouts 
\usepackage{cite}
\usepackage{amsmath,amssymb,amsfonts}
\usepackage{algorithmic}
\usepackage{graphicx}
\usepackage{amsmath}
\usepackage{textcomp}
\usepackage{amsfonts, amsmath}
\usepackage{amssymb}
\usepackage{cleveref}
\usepackage{graphicx}

\newtheorem{theorem}{Theorem}
\newtheorem{lemma}{Lemma}

\newtheorem{assumption}{Assumption}
\newtheorem{property}{Property}
\usepackage{cuted}
\usepackage[english]{babel}
\usepackage{xcolor}
    
    \makeatletter
 \let\old@ps@headings\ps@headings
 \let\old@ps@IEEEtitlepagestyle\ps@IEEEtitlepagestyle
 \def\confheader#1{%
 \def\ps@headings{%
 \old@ps@headings%
 \def\@oddhead{\strut\hfill#1\hfill\strut}%
 \def\@evenhead{\strut\hfill#1\hfill\strut}%
 }%
 \def\ps@IEEEtitlepagestyle{%
 \old@ps@IEEEtitlepagestyle%
 \def\@oddhead{\strut\hfill#1\hfill\strut}%
 \def\@evenhead{\strut\hfill#1\hfill\strut}%
 }%
 \ps@headings%
 }
 \makeatother

 \usepackage[pscoord]{eso-pic}

  \makeatletter
\newcommand{\linebreakand}{%
  \end{@IEEEauthorhalign}
  \hfill\mbox{}\par
  \mbox{}\hfill\begin{@IEEEauthorhalign}
}
\makeatother  

\begin{document}

\title{Adaptive Tracking Control of Uncertain Euler-Lagrange Systems with State and Input Constraints
}

\author{Poulomee~Ghosh and Shubhendu~Bhasin
\thanks{Poulomee Ghosh and Shubhendu Bhasin are with Department of Electrical Engineering, Indian Institute of Technology Delhi, New Delhi, India. 
       {\tt\small (Email: Poulomee.Ghosh@ee.iitd.ac.in, sbhasin@ee.iitd.ac.in)}}}



\maketitle
\begin{abstract}
This paper proposes a novel control architecture for state and input constrained Euler-Lagrange (E-L) systems with parametric uncertainties. A simple saturated controller is strategically coupled with a Barrier Lyapunov Function (BLF) based controller to ensure state and input constraint satisfaction. To the best of the authors' knowledge, this is the first result for E-L systems that guarantee asymptotic tracking with user-specified state and input constraints. The proposed controller also ensures that all the closed-loop signals remain bounded. The efficacy of the proposed controller in terms of constraint satisfaction and tracking performance is verified using simulation on a robot manipulator system. 

\end{abstract}


\section{Introduction}
\label{sec:intro}
 Most practical systems are subjected to constraints in terms of physical limitations, saturation, performance or safety limits which can often be translated into state and input constraints. Conventional adaptive control techniques are only equipped to deal with systems with parametric uncertainties, and ensure boundedness of both plant state and control input, however, the bound is neither known a-priori nor user-defined. In safety-critical applications, maintaining system states within the prescribed bound ensures system safety, albeit at the cost of a higher control effort. Large control magnitude might saturate the actuator and in turn, cause damage or deterioration of the process. Therefore, constraining the
plant states and input within known user-defined bounds
while meeting satisfactory performance objectives is a problem of practical interest.\\
Existing results that account explicitly for state constraints include model predictive control (MPC) \cite{mpc1},\cite{mpc}, optimal control theory \cite{opt1},\cite{opt}, invariant set theory \cite{blanchini},\cite{set}, reference governor approach \cite{rga}, \cite{gilbert} etc. Safety certificates like barrier function (BF), control barrier function (CBF) are widely used to guarantee system safety by ensuring forward invariance of a safe set with respect to a system model \cite{blanchini}. CBF is combined with control Lyapunov function (CLF) by solving a quadratic program (CBF-CLF-QP approach) to develop a stable and safe controller \cite{CBF2}, \cite{CBF3}. Most of these approaches typically require an optimization routine to be run which is computationally expensive and assumes complete model knowledge.\\
Another safety certificate is the Barrier Lyapunov Function (BLF) \cite{BLF},\cite{BLF2} which has been extensively used in literature to satisfy the state constraints for uncertain E-L systems by blending an error transformation \cite{elblf1}. However, most existing BLF-based methods do not address the problem of multiple state constraints that are necessary to ensure the safety of E-L systems. An alternative approach is the state transformation technique using BLF \cite{stc}. Although BLF-based controllers ensure that user-defined state constraints are met, they usually result in large control effort when the states approach
the boundary of the constrained region, often violating the
actuator’s operating limits. Therefore, imposing a user-defined bound on the required control effort along with multiple state  constraints can pave the way for wider applicability for safety-critical systems.\\
Various techniques including saturated functions\cite{incon10, inconnew, incon11, incon12, incon13}, saturated feedback controller\cite{annaswamy}, \cite{Lav}, reference governor, etc. have been extensively used to tackle adaptive control of uncertain plants with user-defined constraint on control input, but the simultaneous satisfaction of constraint on both state and input is still a less explored area of research.\\
Few control techniques exist that deal with the tracking control problem for uncertain nonlinear systems with state and input constraints. MPC \cite{dhar, mpc11, dhar2} is a popular control approach where both state and input constraints can be included in the optimization routine, albeit at the cost of computational complexity. In \cite{zcbf1}, a zeroing control barrier function (ZCBF) is constructed for E-L systems that respect both state and input constraints, but the construction is not straightforward and requires complete model knowledge.\\
The main contribution of this paper is the development of a control methodology for uncertain E-L systems to accommodate user-defined state and input constraints while simultaneously tracking the desired reference trajectory. A saturated feedback controller is designed by considering an auxiliary input signal to constrain the control input which is tactically incorporated with BLF-based controller that deals with the state constraints. Closed-loop signals are guaranteed to be bounded and the trajectory tracking error can be proved to converge to zero asymptotically.


\section{Problem Formulation}
\label{sec:LFSR}
Throughout this paper, $\mathbb{R}$ denotes the set of real numbers, $\mathbb{R}^{p \times q}$ denotes set of $p\times q$ real matrices, the identity matrix in $\mathbb{R}^{p \times p}$ is denoted by $I_{p}$ and $\|.\|$ represents the Euclidian vector norm and corresponding equi-induced matrix norm. $\zeta^{(i)}(t)$ denotes the $i^{th}$ derivative of $\zeta$ with respect to time.

\subsection{Problem Statement}

Consider a general E-L system 
\begin{align}
    M(q)\ddot{q}+V_m(q,\dot{q})\dot{q}+G_r(q)+F_d(\dot{q})=\tau
    \label{plant11}
\end{align}
where $M(q)\in\mathbb{R}^{n\times n}$ denotes a generalized inertia matrix assumed to be known, $V_m(q,\dot{q})\in \mathbb{R}^{n\times n}$ denotes an unknown generalized centripetal-Coriolis matrix,  $G_r(q)\in \mathbb{R}^n$ and $F_d(\dot{q})\in \mathbb{R}^n$ represents the unknown generalized friction and gravity vectors, respectively. $\tau=[\tau_1,\hdots,\tau_n]^T\in\mathbb{R}^n$ is the generalized input control vector, and $q(t), \dot{q}(t), \ddot{q}(t) \in \mathbb{R}^n$ denote the link position, velocity, and acceleration vectors, respectively. \\
The following properties of the system dynamics are used for the subsequent development of control law and stability analysis \cite{Vidyasagar}.
\begin{property}
The inertia matrix $M(q)\in \mathbb{R}^{n \times n}$ is symmetric, positive definite and satisfies the following inequality
\begin{align}
    m_1\|\mu\|^2\leq \mu^TM(q)\mu \leq m_2\|\mu\|^2,
\end{align}
 where $m_1$ and $m_2$ are positive constants and $\mu \in \mathbb{R}^n$ is an arbitrary vector.
\end{property}


\begin{property}
The E-L dynamics is considered to be linearly parameterizable as follows
\begin{align}
    Y(q,\dot{q},\ddot{q})\theta=M(q)\ddot{q}+C(q,\dot{q})\dot{q}+G_r(q)
    \label{Property4}
\end{align}
where $Y:\mathbb{R}^n\times \mathbb{R}^n\times \mathbb{R}^n\rightarrow \mathbb{R}^{n\times m}$ is the known regression matrix, and $\theta\in \mathbb{R}^m$ is the unknown parameter vector.
\end{property}

\textbf{Constraint on plant state:}
 Plant states should remain within a user defined safe set given by $\Omega_q:=\{q(t),\dot{q}(t)\in \mathbb{R}^n:\|q\|< \beta_1, \|\dot{q}\|< \beta_2\}$, where  $\beta_1$ and $\beta_2$ are positive constants.\\
 \textit{Remark 1:} We consider separate bounds for position and velocity states that makes the development in the paper more generalized than imposing a single constraint on the norm of the entire state vector $\|\begin{bmatrix} q^T & \dot{q}^T \end{bmatrix}^T\|$, as typically done in literature.
 
 \begin{assumption}
 The desired reference trajectory $q_d(t)\in \mathbb{R}^n$ and its derivatives $q_d^{(i)}(t)$, $i=1,2$ are known and bounded.
\begin{align}
    & \|q_d(t)\| \leq \alpha_1 < \beta_1\\
    & \|\dot{q}_d(t)\|\leq \alpha_2 < \beta_2
    \end{align}
\end{assumption}
The tracking error dynamics can be expressed as

\begin{align}
    e(t) \triangleq q(t)-q_d(t)
    \label{error}
\end{align}

Provided Assumption 1, the state constraints can be transformed to the constraints on the tracking error states: $\|e(t)\|<\delta_1$, $\|\dot{e}(t)\|<\delta_2$, $\forall t\geq 0$, where $\delta_i \in \mathbb{R}$ are positive constants given by $\delta_i=\beta_i-\alpha_i$, $i=1,2$ i.e. $\|e^{(i-1)}(t)\|<\delta_i \implies\|q^{(i-1)}(t)\|\leq\beta_i$, $i=1,2$.
\\~\\
\textbf{Constraint on control input:} Magnitude of the control input should remain bounded in a safe set given by $\Omega_{\tau} := \{\tau\in\mathbb{R}^{n}:\|\tau(t)\|\leq\tau_{max}\}$, where $\tau_{max}$ is a user-defined positive constant.

\begin{assumption}
For both user-defined state and input constraints imposed on (\ref{plant11}), there exists a feasible control policy $\tau(t)$ that satisfies the control objective. 
\end{assumption}
The control objective is to design an input $\tau(t)$ for the uncertain system (\ref{plant11}), such that $q(t)$ tracks the desired trajectory $q_d(t)$ i.e. $e(t)\rightarrow 0$ and $\dot{e}(t)\rightarrow 0$ as $t \rightarrow \infty$ while both the state and the input remain in user-defined safe set.


%

\section{Proposed Methodology}
To facilitate the design, a filtered  tracking error $r(t)\in \mathbb{R}^n$ is defined as
\begin{align}
    &r=\dot{e}+\alpha e
    \label{fte}
\end{align}
where $\alpha$ is a positive constant. Differentiating (\ref{fte}) and using (\ref{Property4}), the above dynamics can be written as
\begin{align}
  \dot{r}=Y\theta+g\tau  
\end{align}
where, $Y\in\mathbb{R}^{n\times m}$ is known regressor matrix, $\theta\in\mathbb{R}^{m}$ is unknown parameter vector and $g(q)=M^{-1}(q)\in \mathbb{R}^{n\times n}$. Here, $Y\theta$ is given by
\begin{align}
    Y\theta=M^{-1}(-V_m\dot{q}-F_d-G_r-M\ddot{q}_d+M\alpha\dot{e})
\end{align}

\subsection{Input Constraint Satisfaction Using Saturated Control Design}

An auxiliary control input $v(t)\in \mathbb{R}^n$ can be considered as
\begin{align}
&v(t)=g^{-1}(-Y\hat{\theta}-K_1r)
\label{pc1}
\end{align}
where $v(t)\triangleq[v_1(t),\hdots, v_n(t)]^T $, $\hat{\theta}\in\mathbb{R}^m$ is the estimated unknown parameter vector and $K_1(t)\in \mathbb{R}^{n \times n}$ is a positive controller parameter gain. Inspired by \cite{annaswamy}, the saturated feedback controller is designed as
\begin{align}
&\tau_i(t)= \begin{cases}
v_i(t) & \text{if}\:\:\: |v_i(t)|\leq \frac{\tau_{max}}{\sqrt{n}}\\
\frac{\tau_{max}}{\sqrt{n}}sgn(v_i(t)) & \text{if}\:\:\: |v_i(t)|>\frac{\tau_{max}}{\sqrt{n}}
\end{cases},
&& i=1, \hdots , n
\label{pc2}
\end{align}
  
Using (\ref{pc1}) and (\ref{pc2}), the closed-loop dynamics of filtered tracking error can be expressed as
\begin{align}
    \dot{r}=Y\tilde{\theta}-K_1r+g\Delta \tau
    \label{edot}
\end{align}
where $\tilde{\theta}\triangleq\theta-\hat{\theta}\in \mathbb{R}^{m}$ is the parameter estimation error and $\Delta \tau(t)\in\mathbb{R}^{m}$ is defined as the difference between the control input $\tau(t)$ and auxiliary control input $v(t)$, i.e., $\Delta \tau(t) \triangleq \tau(t)-v(t)$. Note that, for the satisfaction of the input constraint, an extra term $g\Delta \tau$ is present in the closed loop dynamics (\ref{edot}), which can be treated as a disturbance term. To mitigate its effect, an auxiliary error signal $r_1(t)\in\mathbb{R}^n$ is considered.
\begin{align}
    \dot{r}_1=-K_1r_1+K_2\Delta \tau
\end{align}
where $K_2(t)\in \mathbb{R}^{n \times n}$ is a time-varying controller parameter. Let $r_d(t)$ be the difference between the actual and auxiliary error signals: $r_d(t) \triangleq r(t)-r_1(t)$, with the following dynamics 
\begin{align}
  \dot{r}_d=Y\tilde{\theta}-K_1r_d+K_d\Delta \tau
\end{align}
where $K_d(t)\triangleq g-K_2(t)\in \mathbb{R}^{n \times n}$ is a controller parameter.\\
\subsection{State Constraint Satisfaction using BLF}\label{AA}

To ensure that the system states are within the user-defined bound, a BLF-based approach \cite{BLF} is considered.

\begin{assumption}
The initial condition of the reference model states $q_d(0), \dot{q}_d(0)$ are chosen such that the initial trajectory tracking error satisfies.
\begin{align}
   &\|e(0)\|<\kappa<\delta_1\\
   &\|\dot{e}(0)\|<\delta_2
\end{align}
where, $\kappa$ is a positive constant, chosen such that $\kappa=\frac{\alpha \delta}{1+\alpha}$ and $\delta=\min\{\delta_1,\delta_2\}$.  
\end{assumption}

For ease of stability analysis, the constraint on the norm of the tracking error and its derivative can be converted into the constraint on the filtered tracking error using (\ref{fte}), as shown subsequently. 




\begin{lemma}
For any positive constant $\kappa$, let $\Omega_r := \{r \in \mathbb{R}^n : \|r\|<\kappa\}\subset \mathbb{R}^n$ and $\Psi:=\mathbb{R}^N\times \Omega_r \subset \mathbb{R}^{N+n}$ be open sets. Consider the system dynamics given by
\begin{align}
    \dot{\mu}=f(t,\mu)
\end{align}
$\mu:=[r^T,\xi^T]^T\in \Psi$, where $\xi$ is the augmentation of the unconstrained states and the function $f:\mathbb{R}_{+}\times \Psi \rightarrow \mathbb{R}^{N+n}$ is measurable  for each fixed $\mu$ and locally Lipschitz in $e$, piecewise continuous and locally integrable on $t$. Suppose, there exists a positive definite, decrescent, quadratic candidate Lyapunov function $V_2(\xi):\mathbb{R}^N\rightarrow\mathbb{R}_{+}$ and continuously differentiable, positive definite, scalar function $V_1(r):\Omega_r \rightarrow \mathbb{R}_{+}$, defined in an open region containing the origin such that 
\begin{align}
    V_1(r)\rightarrow \infty \hspace{10pt} \text{as} \hspace{10pt} \|r\|\rightarrow \kappa
\end{align}
The candidate Lyapunov function can be written as $V(\mu)=V_1(r)+V_2(\xi)$. Given Assumption 3 and $r(0)\in\Omega_r$, if the following inequality holds
\begin{align}
    \dot{V}=\frac{\partial V}{\partial \mu}f\leq 0
\end{align}
then $r(t)\in \Omega_r$  $\forall t$.\\
\begin{proof}
For the proof of Lemma 1, see \cite{BLF}.
\end{proof}

\end{lemma}
To constrain $r(t)$ which in turn ensures state constraint satisfaction, consider a BLF $V_1(r)$ defined on the set $\Omega_r$, such that

\begin{align}
    &V_1(r)\triangleq\frac{1}{2}\log \frac{\kappa^2}{\kappa^{2}-r^Tr}
\end{align}
If $r^Tr\rightarrow \kappa^{2}$, i.e.
when the constrained state $r(t)$ approaches the boundary of the safe set, the BLF $V_1(r)\rightarrow \infty$; this fact will be exploited to guarantee safety of the system.
The unconstrained states involve continuously differentiable and positive-definite quadratic functions. 



Consider the candidate Lyapunov function $V(\mu):\Omega_r\times \mathbb{R}^{N}\rightarrow \mathbb{R}_{+}$ as,
\begin{align}
    V(\mu)&=
    \frac{1}{2}
    \bigg[\log\frac{\kappa^2}{\kappa^{2}-r^Tr}+r_d^Tr_d+\tilde{\theta}^T\Gamma^{-1}\tilde{\theta}\nonumber\\
    &+tr(K_d^T\Gamma_d^{-1}K_d)+tr(K_2^T\Gamma_2^{-1}K_2)\bigg]
    \label{lyap}
\end{align}
where $\mu:=[r^T,r_d^T,\tilde{\theta}^T, K_d^T, K_2^T]^T$.
$\Gamma \in \mathbb{R}^{m \times m}$,  $\Gamma_d \in \mathbb{R}^{n \times n}$ and $\Gamma_2 \in \mathbb{R}^{n \times n}$ are positive-definite matrices. Taking the time-derivative of $V$ along the system trajectory 
\begin{align}
    \dot{V}=&\frac{1}{(\kappa^{2}-r^Tr)}\bigg[r^T(Y\tilde{\theta}-K_1r+g\Delta\tau)\bigg]\nonumber\\
    &+r_d^T(Y\tilde{\theta}-K_1r_d+K_d\Delta \tau)\nonumber\\
    &-\tilde{\theta}^T\Gamma^{-1}\dot{\hat{\theta}}+tr(K_d^T\Gamma_d^{-1}\dot{K}_d)+tr(K_2^T\Gamma_2^{-2}\dot{K}_2) 
    \label{vdot11}
    \end{align}
    Substituting $g=K_d+K_2$ in (\ref{vdot}),
    \begin{align}
    \dot{V}=& \frac{1}{(\kappa^{2}-r^Tr)} \bigg [r^TY\tilde{\theta}-r^TK_1r+r^T(K_d+K_2)\Delta\tau \nonumber\\
    &+ r_d^TY\tilde{\theta}-r_d^TK_1r_d+r_d^TK_d\Delta\tau-\tilde{\theta}^T\Gamma^{-1}\dot{\hat{\theta}} \nonumber\\
    &+ tr(K_d^T\Gamma_d^{-1}\dot{K}_d)+tr(K_2^T\Gamma_2^{-2}\dot{K}_2) 
    \label{vdot}
\end{align}
Adaptive update laws are designed as
\begin{align}
    &\dot{\hat{\theta}}=\bigg[\frac{\Gamma r^TY}{\kappa^{2}-r^Tr}+ \Gamma r_d^TY\bigg]\nonumber\\
    &\dot{K}_d=-\bigg[\frac{\Gamma_d r\Delta\tau^T}{\kappa^{2}-r^Tr}+\Gamma_dr_d\Delta\tau^T \bigg ]\nonumber\\
    &\dot{K}_2=-\frac{\Gamma_2r\Delta \tau^T}{\kappa^{2}-r^Tr}
    \label{proposedlaw11}
\end{align}
Substituting (\ref{proposedlaw11}) in (\ref{vdot11}) yields
\begin{align}
    \dot{V}=-\bigg(\frac{r^TK_1r}{\kappa^{2}-r^Tr}+r_d^TK_1r_d\bigg) \leq 0
    \label{lyapfunc}
\end{align}
which is a negative semi-definite function. 


\begin{theorem}
For the E-L system (\ref{plant11}), provided Assumptions 1-3 hold and the following gain condition is satisfied $0<\alpha < \frac{-1+\sqrt{5}}{2}$, the proposed controller (\ref{pc1}), (\ref{pc2}) and the adaptive laws (\ref{proposedlaw11}) ensure the following.
\begin{enumerate}
    \item[(i)]  The plant states remain within the user-defined safe set given by $\Omega_q:=\{q(t),\dot{q}(t)\in \mathbb{R}^n:\|q\|< \beta_1, \|\dot{q}\|< \beta_2\}$
    \item[(ii)] The control effort is bounded within a user-defined safe set given by $\Omega_{\tau} := \{\tau\in\mathbb{R}^{n}: \|\tau\|\leq \tau_{max}\}$.
    \item[(iii)]  All the closed loop signals remain bounded. 
    \item[(iv)] The trajectory tracking error converges to zero asymptotically i.e. $e(t)\rightarrow 0$ as $t \rightarrow \infty$.
\end{enumerate}
\end{theorem}
\begin{proof}
(i) $V(\mu)$ in (\ref{lyap}) is positive definite and $\dot{V}(\mu)\leq 0$ from (\ref{lyapfunc}), which implies that $V(\mu(t))\leq V(\mu(0))$ $\forall t\geq0$. Since $V(\mu)$ is defined on the region $\Omega_{\mu}:= \{[r^T , \xi^T] \in \Psi: r^Tr\leq \kappa^{2} \}$, it can be inferred from Lemma 1 that 
\begin{align}
  r^Tr< \kappa^{2}  \implies \|r(t)\|< \kappa \hspace{20pt} \forall t\geq0
  \label{lambda1}
\end{align}
Now, by solving the differential equation (\ref{fte}) and employing Assumption 3, it can be proved that 
\begin{align}
    \|e(t)\|< \kappa+\frac{\kappa}{\alpha} \hspace{10pt} \text{and} \hspace{10pt} \|\dot{e}(t)\|<\alpha \delta+\kappa \hspace{15pt} \forall t\geq0
\end{align} 
Substituting the value of $\kappa$ (Assumption 3),
\begin{align}
    &\|e(t)\|<\delta \hspace{64pt} \forall t\geq0\\
    & \|\dot{e}(t)\|<\frac{\delta(\alpha^{2}+2\alpha)}{1+\alpha} \hspace{20pt} \forall t\geq0
    \label{ec}
\end{align}
Now, to show that $\frac{\delta(\alpha^{2}+2\alpha)}{1+\alpha}<\delta$, we choose $\alpha$ as 
\begin{align}
    \alpha^2+\alpha-1<0\implies 0<\alpha<\frac{-1+\sqrt{5}}{2}
\end{align}
which in turn proves that $\|\dot{e}\|<\delta$. Since $\delta=\min\{\delta_1,\delta_2\}$,
\begin{align}
    \|e(t)\|<\delta_1 \hspace{10pt} \text{and} \hspace{10pt} \|\dot{e}(t)\|<\delta_2 \hspace{15pt} \forall t\geq0
    \label{ec12}
\end{align}
 i.e. the trajectory tracking error and its derivative will be constrained within the user-defined safe set : $e(t), \dot{e}(t) \in \Omega_e$ $\forall t\geq0$.\\
Further, since the desired trajectory and the trajectory tracking error are bounded, i.e. $\|q_d(t)\|\leq \alpha_1$, $\|e(t)\|< \delta_1$, $\|\dot{e}(t)\|< \delta_2$, it can be easily inferred from (\ref{error}) that the proposed controller ensures the plant states are bounded within the user defined safe set
\begin{subequations}
\begin{align}
    &\|q(t)\|<\delta_1+\alpha_1= \beta_1 \\
     &\|\dot{q}(t)\|<\delta_2+\alpha_2= \beta_2&& \forall t \geq 0
\end{align}
\end{subequations}
(ii) The control effort of the proposed controller $\tau(t)=[\tau_1(t),\hdots, \tau_n(t)]^T$ and $\|\tau(t)\|=\sqrt{\tau_1^2(t)+\tau_2^2(t)+\hdots+\tau_n^2(t)}$. For constraining the control input two cases are considered.\\
\textit{Case 1:} $\|v_i(t)\|\leq \frac{\tau_{max}}{\sqrt{n}}$\\
For this case, $\tau_i(t)=v_i(t)$ and $\Delta \tau(t)=0$.
So, $|\tau_i|\leq \frac{\tau_{max}}{\sqrt{n}}$ which implies $\|\tau\|<\tau_{max}$\\
\textit{Case 2:} $\|v_i(t)\|> \frac{\tau_{max}}{\sqrt{n}}$\\
For this case, $\tau_i(t)=\frac{\tau_{max}}{\sqrt{n}}sgn(v_i(t))$ which proves $\|\tau\|<\tau_{max}$.\\

(iii) Since the closed-loop trajectory tracking error and the controller parameter estimation errors remain bounded and $\theta(t)$ is constant, it can be concluded that the estimated parameters are also bounded i.e. $\hat{\theta}(t)\in \mathcal{L}_{\infty}$ followed by ensuring the plant states $q(t), \dot{q}(t)$ and control input $\tau(t)$ to be bounded for all time instances. Thus, the proposed controller guarantees that all the closed-loop signals are bounded.\\


(iv) Since $V(\mu)>0$ and $\dot{V}(\mu)$ is negative semi-definite (\ref{lyapfunc}), it can be shown that 
$r(t)$, $r_d(t)$, $\tilde{\theta}(t)$, $Y(t)$, $K_d(t)$, $K_1(t)$, $K_2(t)$, $\Delta\tau(t)$ $\in \mathcal{L}_{\infty}$. Further, from (\ref{lyapfunc}) it can be shown that $r(t)\in \mathcal{L}_2$ and from (\ref{edot}) it can be inferred that $\dot{r}(t)\in \mathcal{L}_{\infty}$. Therefore, $r(t)$ is uniformly continuous. Invoking Barbalat's Lemma \cite{slotine}, it can be proved that $r(t)\rightarrow 0$ as $t\rightarrow\infty$, which consequently ensures that $e(t)$ and $\dot{e}(t)$ converges to zero asymptotically as $t \rightarrow \infty$.
\end{proof}

\section{Simulation Results}
To demonstrate the efficacy of the proposed algorithm for constrained E-L systems, the dynamics of a two-link robot manipulator is considered.

\begin{equation}
M(q)\ddot{q}+V_m(q,\dot{q})\dot{q}+F_d(\dot{q})+G_r(q)=\tau   
\label{plant}
\end{equation}
where the respective matrices related to (\ref{plant}) are given by:
\begin{align*}
    &M(q)=\begin{bmatrix}
    p_1+2p_3c_2 & p_2+p_3c_2\\
    p_2+p_3c_2 & p_2
    \end{bmatrix}\\
    &V_m(q,\dot{q})=\begin{bmatrix}
    -p_3s_2\dot{q_2} & -p_3s_2(\dot{q_1}+\dot{q_2})\\
    p_3s_2\dot{q_1}  & 0
    \end{bmatrix}\\
    &F_d(\dot{q})=\begin{bmatrix}
    f_{d_1} & 0\\
    0 & f_{d_2}
    \end{bmatrix}\\
    &G_r(q)=0_{2\times 1}
\end{align*}
where, $q(t)=[q_1(t), q_2(t)]^T\in \mathbb{R}^2$ and $\dot{q}(t)=[\dot{q}_1(t), \dot{q}_2(t)]^T\in \mathbb{R}$ denote the angular position (rad) and the velocity of the robot links respectively.
The desired trajectory is considered as
\begin{equation*}
   q_d(t)=\begin{bmatrix}
    0.5sint \\ 2cos(t/4)
    \end{bmatrix}
\end{equation*}
The control objective is to design a controller such that plant states $q(t),\dot{q}(t)$  track the desired reference trajectory $q_d(t),\dot{q}_d(t)$, while simultaneously satisfying the state and input constraints given by
\begin{align*}
    &\|q(t)\|<\beta_1=3.6\\
    &\|\dot{q}(t)\|<\beta_2=2.1\\
    &\|\tau(t)\|\leq\tau_{max}=5
\end{align*}
The other parameters used for simulation are chosen as: $\Gamma=10\mathbb{I}_{m\times m}$, $\Gamma_d=5\mathbb{I}_{n\times n}$, $\Gamma_2=5\mathbb{I}_{n\times n}$, 
$\alpha_1=2$, $\alpha_2=0.6$, $\alpha=0.6$, $K_1=10$, $p_1=3.473$ kg-m, $p_2=0.196$ kg-m, $p_3=0.242$ kg-m, $f_{d_1}=5.3$ N s, $f_{d_2}=1.1$ N s.\\
Given Assumption 1, $\|q_d\| \leq 2$, $\|\dot{q}_d\| \leq 0.6$, the state constraint is equivalent to satisfying the constraint on the trajectory tracking error i.e. $\|e\|<\delta_1=1.6$, $\|\dot{e}\|<\delta_2=1.5$. As $\delta=\min\{\delta_1,\delta_2\}=1.5$ and $\kappa=\frac{\alpha \delta}{1+\alpha}=0.56$, from proof of Theorem 1 it can be easily inferred that satisfying constraint on the filtered tracking error $\|r\|<0.56$ will consequently ensure state constraint satisfaction.\\
To show the effectiveness of the proposed control law, we compare it with classical adaptive controller where the control input and adaptive update law are given by
\begin{align}
&\tau_{c}(t)=g^{-1}(-Y\hat{\theta_c}-K_1r)\\
&\dot{\hat{\theta}}_c=\Gamma_cr^TY
\label{c1}
\end{align}
where the adaptive gain is chosen as $\Gamma_c=100\mathbb{I}_{m\times m}$. Note that, adaptation gains for both the proposed controller and the classical method are tuned to achieve comparable tracking performance. 
\begin{figure}[htbp!]
\centering
\includegraphics[width=\linewidth]{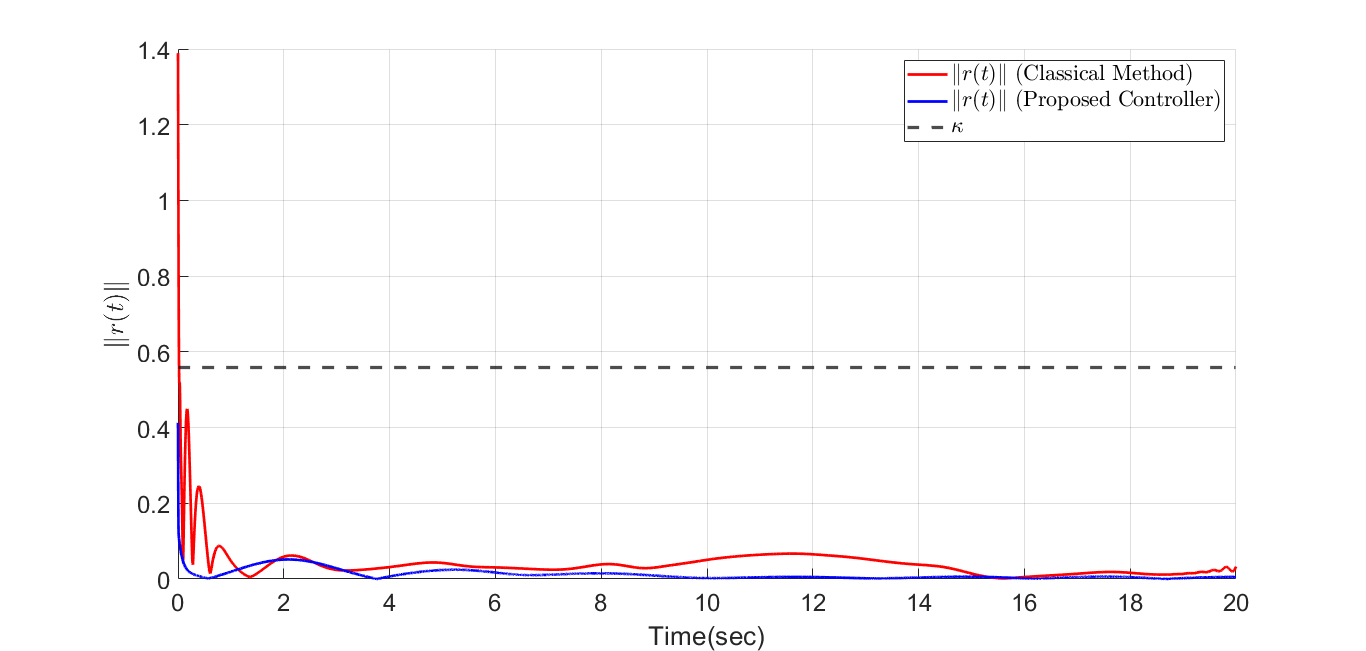}
\caption{Comparison of the filtered tracking error between the classical (\ref{c1}) and the proposed (\ref{proposedlaw11}) adaptive controllers.}
\label{fte111}
\end{figure}

\begin{figure}[htbp!]
\centering
\includegraphics[width=\linewidth]{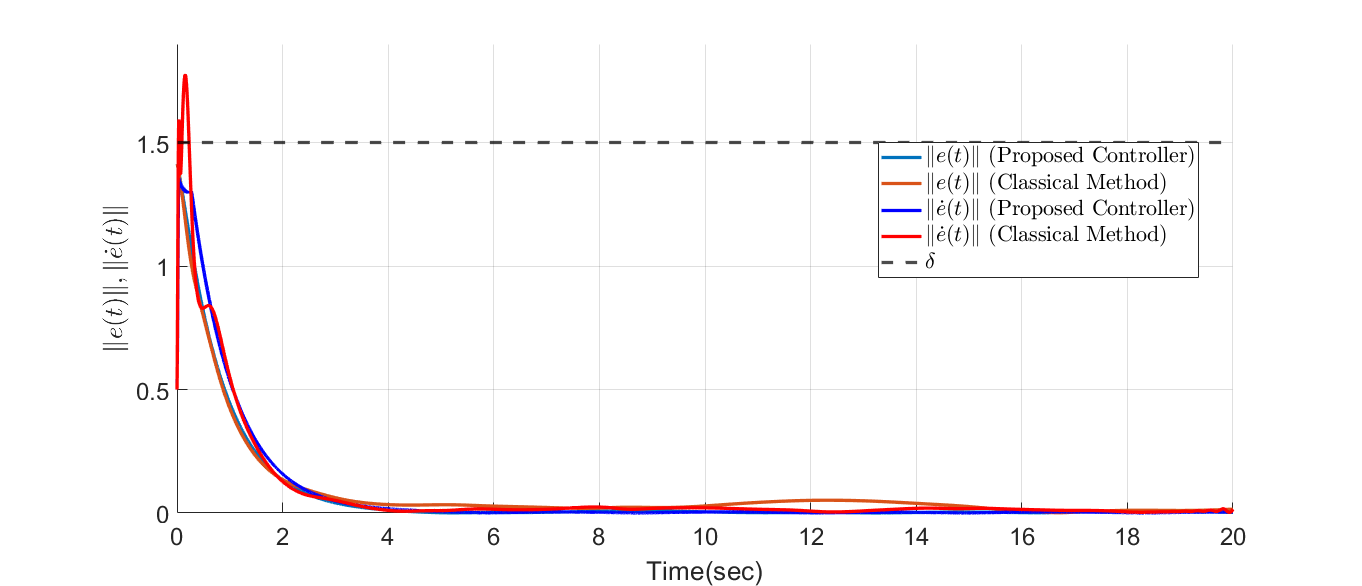}
\caption{Comparison of the trajectory tracking error between the classical (\ref{c1}) and the proposed (\ref{proposedlaw11}) adaptive controllers.}
\label{e}
\end{figure}

\begin{figure}[htbp!]
\centering
\includegraphics[width=\linewidth]{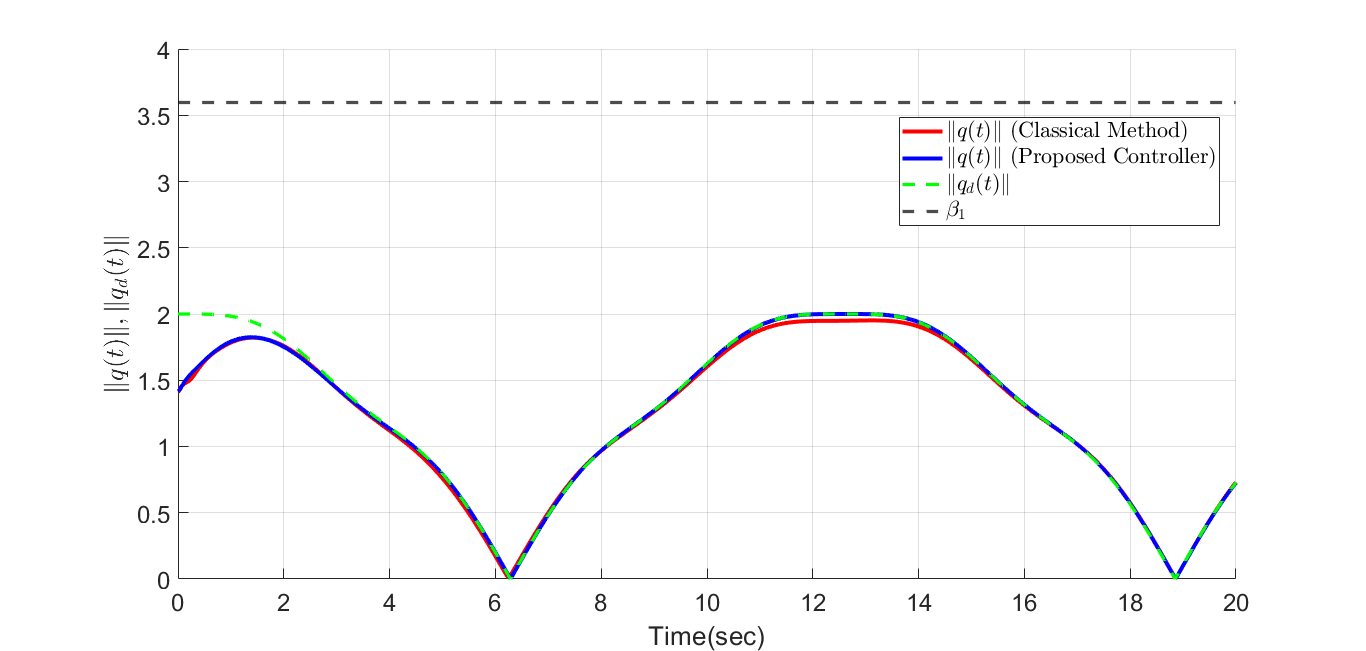}
\caption{Comparison of the angular position $({q}(t))$  between the classical (\ref{c1}) and the proposed (\ref{proposedlaw11}) adaptive controllers. The desired velocity and the velocity constraint are shown using the dotted lines.}
\label{q}
\end{figure}

\begin{figure}[htbp!]
\centering
\includegraphics[width=\linewidth]{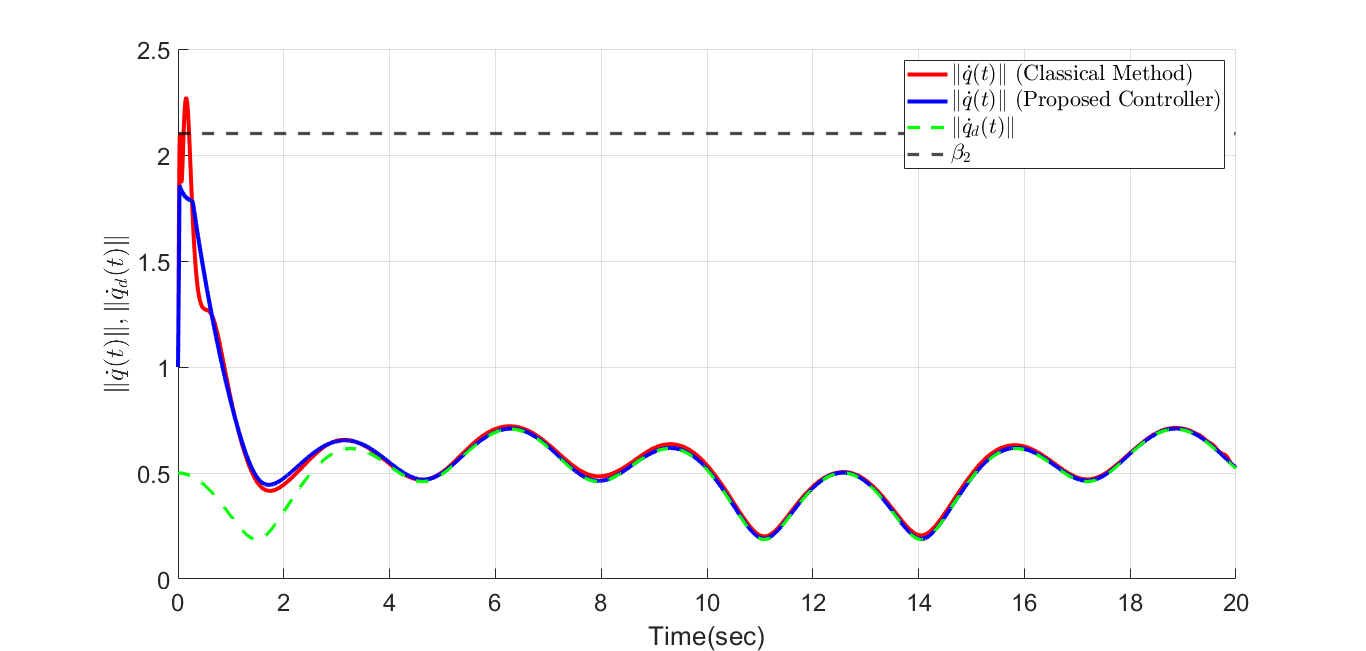}
\caption{Comparison of the angular velocity $(\dot{q}(t))$ between the classical (\ref{c1}) and the proposed (\ref{proposedlaw11}) adaptive controllers. The desired velocity and the velocity constraint are shown using the dotted lines.}
\label{qdot}
\end{figure}

Fig. \ref{fte111} shows that the proposed law ensures the filtered tracking error remains within bounds, i.e. $\|r\|<0.56$, which in turn ensures the boundedness of the trajectory tracking error and its derivative within user-defined constraints, as seen in Fig. \ref{e}, while the constraints are violated with the classical method. Figs. \ref{q}-\ref{qdot} show that the proposed controller guarantees that the states are within the user-defined safe set while tracking the desired reference trajectory, however states go beyond the safe region using the classical method. Furthermore, control effort is also confined to the user-defined constrained region with the proposed control technique, whereas the input constraint gets violated in case of the classical approach (Fig. \ref{input111}). 
\begin{figure}[h!]
\centering
\includegraphics[width=\linewidth]{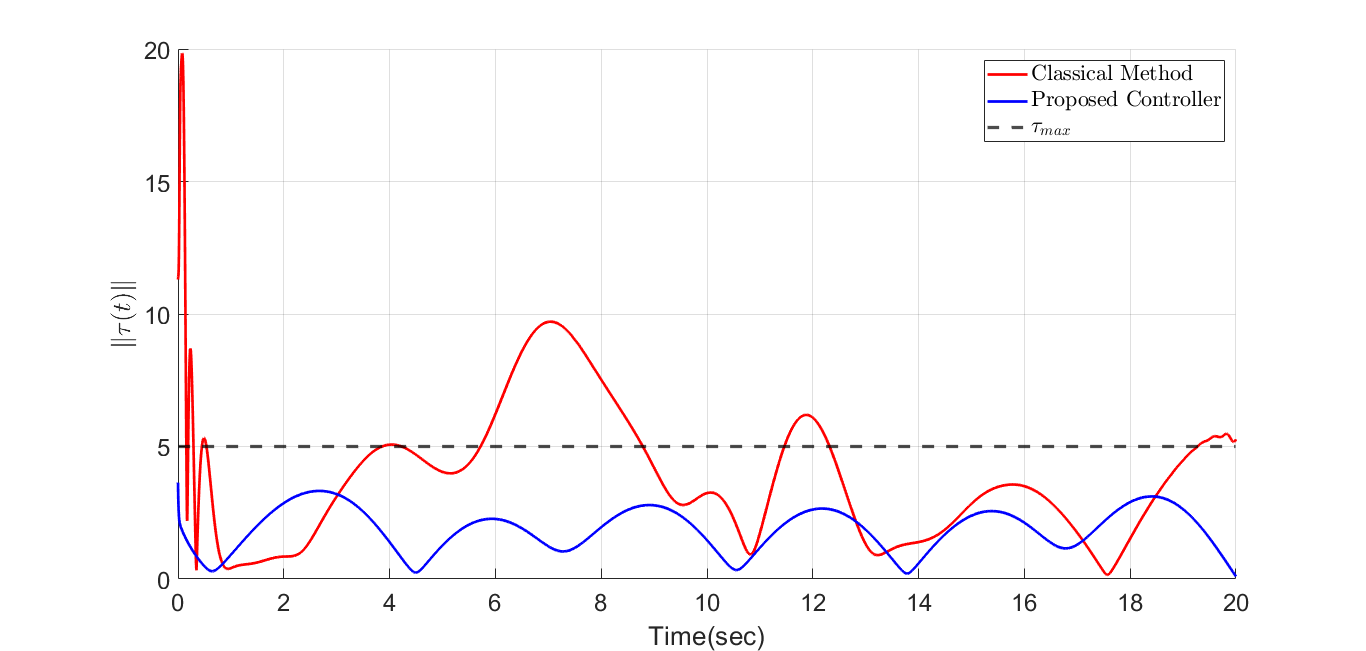}
\caption{Comparison of the control input ($\tau(t)$) between the classical (\ref{c1}) and the proposed (\ref{proposedlaw11}) adaptive controllers.}
\label{input111}
\end{figure}





 It is seen that increasing the adaptation gain leads to better tracking performance, although the response becomes more oscillatory. For the case of conventional adaptive controller, the improved tracking performance is achieved at the cost of greater control amplitude, leading to violation of the input constraints. The high frequency oscillations in the control input may even violate the actuation rate limits. In contrast, the proposed controller guarantees that the control input remains bounded within the user-defined safe sets for all future time while also ensuring the pre-specified bounds on the plant states.\\
\textit{Remark:} It may be possible to meet the state and input constraints for classical adaptive control by tuning the adaptation gains, however, it is seen that such scenarios are almost always accompanied by degradation in tracking performance. This trade-off between performance and constraint satisfaction is typical of classical adaptive controllers; the proposed approach attempts to address such concerns and enhance the practical applicability of adaptive control for safety-critical applications.

\section{Conclusion}
In this paper, an adaptive control method is proposed for uncertain E-L systems  with user-defined state and input constraints. A BLF-based controller is strategically combined with a saturated controller to ensure that  both the plant state and the control input remain bounded within user-defined safe sets while tracking a desired reference trajectory. The proposed controller also guarantees that the trajectory tracking error asymptotically converges to zero and all the closed-loop signals remain bounded. Simulation studies validate the efficacy of the proposed control law.
Extending the work to other classes of nonlinear systems and analyzing the robustness properties are important areas of future research.

\bibliographystyle{ieeetr}
\bibliography{root}
\end{document}